\documentclass{article}
\pdfoutput=1
\usepackage{amsmath,amsthm,amsfonts}
  \newtheorem{definition}{Definition}
  \newtheorem{theorem}{Theorem}
  \newtheorem{lemma}{Lemma}
  \newtheorem{remark}{Remark}
  \newtheorem{corollary}{Corollary}
\newcommand{\QUESTION}[1]{}

\usepackage{graphicx,url,multirow}
\usepackage{algorithm,algorithmic}

\renewcommand{\algorithmicif}{\textbf{If}}
\renewcommand{\algorithmicthen}{\textbf{then}}
\renewcommand{\algorithmicelse}{\textbf{else}}

\makeatletter

  \newcommand{\IFTHEN}[3][default]{\ALC@it\algorithmicif\ #2\
    \algorithmicthen\ #3\
    \ifthenelse{\boolean{ALC@noend}}{}{\algorithmicendif\ } \ALC@com{#1}}
  \newcommand{\IFTHENNOEND}[3][default]{\ALC@it\algorithmicif\ #2\
    \algorithmicthen\ #3\
    \ALC@com{#1}}
  \newcommand{\ONLYELSE}[1][default]{\ALC@it\algorithmicelse%
\ALC@com{#1}\begin{ALC@if}}
\makeatother

\newcommand{\N}{{\ensuremath{\mathbb N}}}
\newcommand{\T}{{\ensuremath{\mathcal T}}}
\newcommand{\jspace}[1]{}
\newcommand{\keywords}[1]{}

\urldef\jgdemail\url{Jean-Guillaume.Dumas@imag.fr}
\urldef\hhemail\url{Hossayni.Hicham@gmail.com}

\newcommand{\mytitle}{Matrix powers algorithms for trust evaluation in public-key infrastructures}
\title{\mytitle}
\newcommand{\myauthor}{Jean-Guillaume Dumas and Hicham Hossayni}
\author{Jean-Guillaume Dumas\thanks{Laboratoire J. Kuntzmann, Universit\'e de
Grenoble. 51, rue des Math\'ematiques, umr CNRS 5224, bp 53X, F38041
Grenoble, France, \jgdemail. }
\and Hicham Hossayni\thanks{CEA / L\'eti, 17 rue des Martyrs, 38054 Grenoble, France, \hhemail}}

 \usepackage{color,svgcolor}
 \makeatletter
 \usepackage{hyperref}
  \hypersetup{
  pdftitle={\mytitle},
  pdfauthor={\myauthor},
  breaklinks=true,
  colorlinks=true,
  linkcolor=blue,
  citecolor=darkred,
  urlcolor=darkblue,
  }
 \makeatother

\begin{document}
\maketitle

\begin{abstract}
This paper deals with the evaluation of trust in public-key infrastructures.
Different trust models have been proposed to interconnect the various PKI
components in order to propagate the trust between them.
In this paper we provide a new polynomial algorithm using linear algebra to
assess trust relationships in a network using different trust evaluation
schemes. 
The advantages are twofold: first the use of matrix computations instead of
graph algorithms provides an optimized computational solution; second, our
algorithm can be used for generic graphs, even in the presence of cycles. 
Our algorithm is designed to evaluate the trust using all existing (finite)
trust paths between entities as a preliminary to any exchanges between
PKIs. This can give a precise evaluation of trust, and accelerate for instance cross-certificate validation.
\end{abstract}

\keywords{Trust evaluation, Matrix powers, Spectral analysis of networks, Distributed PKI trust model.}

\section{Introduction}
The principle of a Public Keys Infrastructure (PKI) is to establish (using
certificates) a trust environment between network entities and thus guarantee
some security of communications.

For instance companies can establish hierarchical PKI's with a certification
authority (CA) signing all the certificates of their employees. In such a
setting there is a full trust between employees and their CA. Now if entities
with different PKI structures want to communicate, either they find a fully
certified path between them or they don't. In the latter case some degree of
trust has to be established between some disjoint entities.

Ellison and Schneier identified a risk of PKIs to be ``Who do
we trust, and for what?'' which emphasizes the doubts about the trust
relationship between the different PKI components
\cite{Ellison:2000:tenrisks}.
Several incidents, including the one in which VeriSign issued to a
fraudulent two certificates associated with Microsoft
\cite{Gomes:2001:verisign}, or even the recent fraudulent certificates for
Google emitted by DigiNotar \cite{Adkins:2011:diginotar}, confirm
the importance of the trust relationship in the trust models, see
also~\cite{Josang:2011:STM}.
This leads to the
need of a precise and a global evaluation of trust in PKI architectures.
Another approach would be to use some fully trusted keys or authorities,
like the Sovereign Keys or the Convergence
project\footnote{\url{https://www.eff.org/sovereign-keys},
  \url{http://convergence.io}}, or e.g. trust
lists~\cite{Rifa-Pous:2007:trustlists}.

For example in a cross-certification PKI, an entity called
Alice can establish a communication with another entity called Bob
only after validating Bob's certificate. For this, Alice must
verify the existence of a certification path between her trust anchor
and Bob's certification authority (CA).
This certificate validation policy imposes that each entity must have a
complete trust in their trust anchors, and that this trust anchor has a
complete (direct or indirect) trust relationship with other CAs.

In e.g. \cite{Guha:2004:PTD,Huang:2009:CTA,Huang:2010:FCT,Foley:2010:stm}
algorithms are proposed to quantify the trust
relationship between two entities in a network, using transitivity.
Some of them evaluate trust throughout a single path, while others
consider more than one path to give a better approximation of trust
between entities. However to the best of our knowledge they are
restricted to simple network {\em trees}.
In this paper we choose the last approach and use transitivity
to efficiently approximate global trust degree.
Our idea is to use an adapted power of the adjacency
matrix (used e.g. to verify the graph connectivity or to compute the
number of {\em finite} paths between nodes). This spectral approach is
similar to that used also e.g. for community detection in graphs
\cite{Estrada:2009:community} and we use it to produce a centralized
or distributed quantification of trust in a network. Moreover it
allows to deal with {\em any kind of graphs, without any restrictions
to trees nor dags}.

More generally, the aim of this paper is to propose a generic solution
for trust  evaluation adapted to different trust models.
The advantage of spectral analysis is twofold :
first the use of matrix computations instead of graph algorithms provides an
optimized computational solution in which the matrix memory management is more
adapted to the memory hierarchy; second, our algorithm can be used for
generic graphs, even in the presence of cycles.
Moreover, our algorithm is iterative so that a good approximation of the
global trust degrees can be quickly computed: this is done by fixing  the
maximum length of the trust paths that are considered.
The complexity of this algorithm is $O(n^3\cdot\varphi\cdot \ell)$ in the worst
case, polynomial in $n$, the number of entities (nodes of the graph), $\varphi$,
the number of trust relationships (edges), and $\ell$, the size of the longest
path between entities.
For instance the algorithm proposed in \cite{Huang:2009:CTA}
worked only for directed acyclic graphs (DAG) and required the approximate
resolution of the Bounded Disjoint Paths problem, known to be NP-Hard
\cite{Reiter:1998:RAU}. In case of DAGs the complexity of our
algorithm even reduces to $O(n\cdot\varphi\cdot \ell)$.

Our algorithm is designed to evaluate the trust using
all existing (finite) trust paths between entities as a preliminary
to any exchanges between PKIs.
This can give a precise evaluation of trust, and optimize the
certificate validation time.
These computations can be made, in a centralized manner by a trusty
independent entity, like Wotsap for a web of trust networks\footnote{\url{http://www.lysator.liu.se/~jc/wotsap/index.html}}, 
by CAs in the case of cross-certification PKI (e.g. via PKI Resource
Query Protocols~\cite{Pala:2008:Peaches}), or by the users themselves in
the case of PGP web of trust.
The latter can even happen in a distributed
manner~\cite{Dolev:2010:CMT} or with
collaborations~\cite{Pala:2010:PKS}.

Our algorithm works for generic trust metrics but can be more efficient when the
metrics form a ring so that block matrix algorithms can be used.
We thus present in section \ref{sec:metric} different possible trust metrics and
their aggregation in trust network.
We then present the transformation of DAG algorithms for the computation of
the aggregation into matrix algorithms in section \ref{sec:powers}.
We do this in the most generic setting, i.e. when even the dotproducts have to
be modified, and with the most generic trust metric (i.e. including trust,
uncertainty {\em and verified distrust}).
Finally, in section \ref{sec:cycles}, we present our new polynomial algorithm for
generic graphs.

\section{Transitive trust metrics}\label{sec:metric}
\subsection{The calculus of trust}
There are several schemes for evaluating the (transitive)
trust in a network. Some present the trust degree as a single value
representing the probability that the expected action will happen. The
complementary probability being an uncertainty on the trust.

Others include the
\emph{distrust} degree indicating the probability that the opposite of
the expected action will happen \cite{Guha:2004:PTD}.
More complete schemes can be introduced to
evaluate trust: J{\o}sang \cite{Josang:2007:PLUU} for instance introduced the
Subjective Logic notion which
expresses subjective beliefs about the truth of propositions with
degrees of "uncertainty".

\cite{Huang:2009:CTA,Huang:2010:FCT} also introduced a quite similar scheme
with a formal, semantics based, calculus of trust and applied it to
public key infrastructures (PKI).
We chose to present this metric for its generality and include in this section
some definitions and theorems taken from \cite{Huang:2009:CTA}.
The idea is
to represent trust by a triplet, (trust, distrust, uncertainty).
Trust is the proportion of experiences proved, or believed, positive. Distrust
is the proportion of experiences {\em proved negative}. Uncertainty is the proportion
of experiences with unknown character.

\begin{definition}[{\cite[\S 5.1]{Huang:2009:CTA}}]
Let $d$ be a trustor entity and $e$ a trustee.
Let $m$ be the total number of encounters between $d$ and $e$ in a given
context.
Let $n$ (resp. $l$) be the number of positive (resp. negative)
experiences among all encounters between $d$ and $e$.
\begin{itemize}
\item {\bf The trust degree} is defined as the frequency rate of
the trustor's positive experience among all encounters with the
trustee. That is,
$td(d,e)=\frac{n}{m}$.
\item {\bf The distrust degree}:  similarly we have
$dtd(d,e)=\frac{l}{m}$.
\item {\bf The uncertainty}: denoted by \textit{ud} is defined by:
$ud(d,e)=1-td(d,e)-dtd(d,e)$.
\end{itemize}
\end{definition}
In the following we will denote the {\bf trust relationship} by a triple
$tr(a,b)=\langle td(a,b)$, $dtd(a,b)$, $ud(a,b)\rangle$ or simply
$tr(a,b)=\langle td(a,b),dtd(a,b)\rangle$ since the uncertainty is completely determined by
the trust and distrust degrees.

In these definitions, the trust depends on the kind of expectancy, the
context of the experiences, type of trust (trust in belief,  trust in
performance), ..., see e.g. \cite{Huang:2010:FCT}.
For simplicity,
we only consider in the next sections the above generic concept of
trust.

\subsection{Aggregation of trust}
The main property we would like to express is {\em transitivity}. Indeed
in that case keys trusted by many entities, themselves highly trusted,
will induce a larger confidence.
In the following we will consider a trust graph representing the trust
relationships as triplets between entities in a network.
\begin{definition}[Trust graph]
Let $\T \subset [0,1]^3$ be a set of trust relationships.
Let $V$ be a set of entities of a trust network.
Let $E$ be a set of directed edges with weight in $\T$.
Then $G=(V,E,\T )$ is called a {\em trust graph} and there
is an edge between two vertices whenever there exist a
nonzero trust relationship between its entities.
\end{definition}

Next we define the transitivity over a path between entities and using
parallel path between them as sequential and parallel aggregations.
We first need to define a trust path:
\begin{definition}[Trust path] Let $G=(V,E,\T)$ be a trust graph.
A trust path between two entities $A_1 \in V$ and $A_n \in V$ is
defined as the chain, $A_1\overset{t_1}{\longrightarrow}
A_2\overset{t_2}{\longrightarrow}
... A_{n-1}\overset{t_{n-1}}{\longrightarrow} A_n$, where $A_i$ are
entities in $V$ and $t_i \in \T$ are respectively the trust degrees
associated to each trust relation
$(A_i\overset{t_i}{\longrightarrow}A_{i+1}) \in E$.
\end{definition}

The need of the sequential aggregation is shown by the following example.
Consider, as shown on figure \ref{fig:SeqAggr}, Alice trusting Bob with a certain
degree, and Bob trusting Charlie with a certain trust degree.
Now, if Alice wishes to communicate with Charlie, how can she evaluate her
trust degree toward him?
For this, we use the sequential aggregation of trust to help Alice to
make a decision,
and that is based on Bob's opinion about Charlie.

\begin{figure}[htbp]
\begin{minipage}{0.3\textwidth}\center
    \includegraphics*[width=0.9\textwidth, keepaspectratio=true]{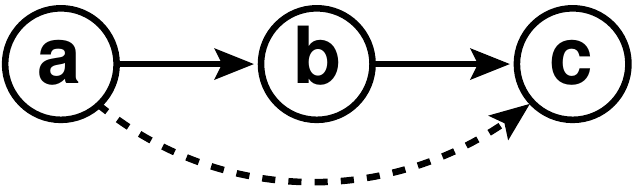}
\end{minipage}\hfill
\begin{minipage}{0.7\textwidth}
\begin{itemize}
\item[$\longrightarrow$]: Direct trust relationship
\item[$\dashrightarrow$]: Indirect (sequentially aggregated) trust relationship
\end{itemize}
\end{minipage}
    \caption{Simple sequential trust aggregation}\label{fig:SeqAggr}
\end{figure}

\begin{definition}[Sequential aggregation of
    trust {\cite[Theorem UT-1]{Huang:2009:CTA}}]\label{def:seqprop}
Let $G=(V,E,\T)$ be a trust graph.
Let $a$,$b$ and $c$ be three entities in $V$ and $tr(a,b)\in \T$,
$tr(b,c) \in \T$ be respectively the trust degrees associated to the
entity pairs $(a,b)$ and $(b,c)$. The {\em sequential aggregation of
  trust} between $a$ and $c$ is a function $f$, that calculates the
trust degree over the trust path $a\rightarrow b \rightarrow c$. It is
defined by :
\begin{align*}
f : \T \times \T \rightarrow \T& ~\text{\em with}~& f(tr(a,b),tr(b,c))&=tr_f(a,c)= \langle td_f(a, c),dtd_f(a, c)\rangle\\
  &\text{\em where}~~& td_f(a, c)&= td(a, b).td(b, c) + dtd(a, b).dtd(b, c)\\
  && dtd_f(a, c)&= dtd(a, b).td(b, c) + td(a, b).dtd(b, c)\end{align*}
\end{definition}

\begin{lemma}\label{lem:seqsound}
Definition \ref{def:seqprop} is sound. 
\end{lemma}
\begin{proof}
We consider the trust path $a\overset{x}{\to} b \overset{y}{\to}c$ and let
$x=\langle x_t,x_d,x_u \rangle$, $y=\langle y_t,y_d,y_u\rangle$ and $z=\langle
z_t,z_d,z_u\rangle=f(x,y)$. 
Then
$z_u=1-x_ty_t-x_dy_d-x_ty_d-x_dy_t=1-x_t(y_t+y_d)-x_d(y_t+y_d)=1-(x_t+x_d)(y_t+y_d)=1-(1-x_u)(1-y_u)$.
Since $0\leq x_u\leq 1$ and $0<leq y_u\leq 1$, we have that 
$0\leq (1-x_u)(1-y_u) \leq 1$ so that $0\leq z_u=x_u+y_u\leq 1$. Since $x_u$ and
$y_u$ are positive by definition, it follows that $f(x,y)$ {\em is} a trust
relationship.
\end{proof}

From the definition we also for instance immediately get the following
properties.
\begin{lemma}\label{lem:sequnincr}
$f$ (Sequential aggregation) increases uncertainty.
\end{lemma}
\begin{proof}
As in the proof of lemma \ref{lem:seqsound}, we let $x=\langle x_t,x_d,x_u \rangle$, $y=\langle y_t,y_d,y_u\rangle$ and
$z=f(x,y)$. From $z_u=1-(1-x_u)(1-y_u)$ we have that $z_u=x_u+y_u(1-x_u)$.
Therefore as $y_u$ and $1-x_u$ are positive we in turn have that $z_u\geq x_u$.
We also have $z_u=y_u+x_u(1-y_u)$ and therefore also $z_u\geq y_u$. Moreover we
see that if both uncertainties are non zero, then the uncertainty of $f(x,y)$ is
strictly increasing.
\end{proof}
\begin{lemma}[{\cite[Property 2]{Huang:2010:FCT}}]
$f$ (Sequential aggregation) is associative
\end{lemma}
\begin{proof}
We consider the trust path $a\overset{x}{\to} b \overset{y}{\to}
c\overset{z}{\to} d$ and let $x=\langle x_t,x_d\rangle$, $y=\langle y_t,y_d\rangle$, $z=\langle z_t,z_d\rangle$.
Then $f(x,y)=\langle x_ty_t+x_dy_d,x_dy_t+x_ty_d\rangle$ and
$f(y,z)=\langle y_tz_t+y_dz_d,y_dz_t+y_tz_d\rangle$, so that 
$f(f(x,y),z)=\langle
(x_ty_t+x_dy_d)z_t+(x_dy_t+x_ty_d)z_d,(x_dy_t+x_ty_d)z_t+(x_ty_t+x_dy_d)z_d\rangle$.
In other words, $f(f(x,y),z)=\langle
x_t(y_tz_t+y_dz_d)+x_d(y_dz_t+y_tz_d),x_d(y_tz_t+y_dz_d)+x_t(y_dz_t+y_dz_d)\rangle=f(x,f(y,z))$.
\end{proof}
From this associativity, this sequential aggregation function can be applied
recursively to any tuple of values of $\T$, to evaluate the sequential
aggregation of trust over any trust path with any length $\ge 2$, for instance
as follows: $f(t_1, ..., t_n) = f(f(t_1, ..., t_{n-1}), t_n)$.

Now, the following definition of the parallel aggregation function can also be
found in {\cite[$\S$ 7.2.2]{Huang:2010:FCT}}, it is clearly associative and is
illustrated on figure \ref{fig:ParAggr}.  
\begin{definition}[Parallel aggregation of trust~{\cite[\S 6.2]{Huang:2009:CTA}}]\label{def:parprop}
Let $G=(V,E,\T)$ be a trust graph.
Let $a, b_1,\ldots, b_n, c$ be entities in $V$
and $tr_i(a,c)\in\T$ be the trust degree over the
trust path $a \rightarrow b_i \rightarrow c$ for all $i\in {1..n}$.
The parallel aggregation of trust is a function $g$, that calculates
the trust degree associated to a set of disjoint trust paths
connecting the entity a to the entity c. It is defined by:
\begin{align*}
    g: \T^n \rightarrow \T& ~\text{\em with}~ g([tr_1,tr_2,\ldots,tr_n](a,c))=tr_g(a,c) = \langle td_g(a,c),dtd_g(a,c)\rangle\\
  \text{\em where}~~& td_g(a,c)=1-\prod_{i=1..n}{(1-td_i)} ~\text{\em and}~ dtd_g(a,c)=\prod_{i=1..n}{dtd_i}\end{align*}
\end{definition}
\begin{figure}[htbp]
\begin{minipage}{0.3\textwidth}\center
    \includegraphics*[width=0.8\textwidth, keepaspectratio=true]{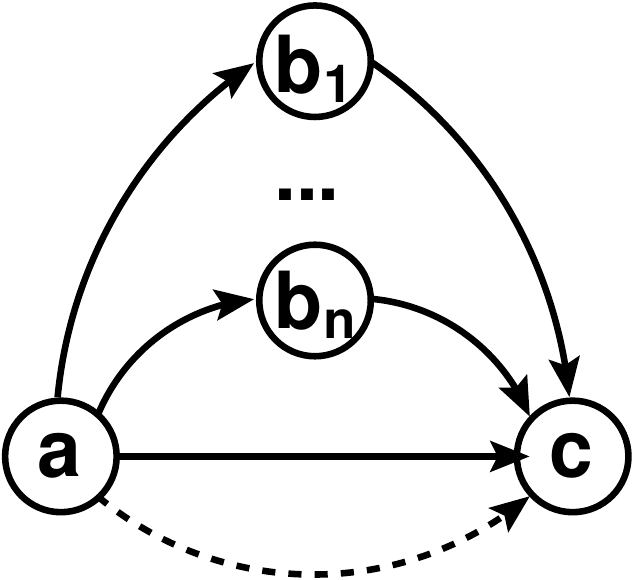}
\end{minipage}\hfill
\begin{minipage}{0.7\textwidth}
\begin{itemize}
\item[$\longrightarrow$]: Direct trust relationship
\item[$\dashrightarrow$]: Indirect (parallely aggregated) trust relationship
\end{itemize}
    \caption{Parallel aggregation of trust for multiple trust}\label{fig:ParAggr}
\end{minipage}
\end{figure}
\begin{lemma}Definition \ref{def:parprop} is sound.
\end{lemma}
\begin{proof}
We start by proving it for two paths of length $1$. 
Let $x=\langle x_t,x_d\rangle$ and $y=\langle y_t,y_d\rangle$. Then
$g(x,y)=\langle 1-(1-x_t)(1-y_t),x_dy_d\rangle$. It is clear that $td_g$ and $dtd_g$ are non
negative. Now from the definition of trust relationship we know that $x_t+x_d
\leq 1$ and $y_t+y_d \leq 1$ so that $x_d \leq 1-x_t$ and $y_d \leq
y_t-1$. 
Therefore $x_dy_d \leq (1-x_t)(1-y_t)$ and $td_g(x,y) + dtd_g(x,y) \leq 1$. This
generalizes smoothly to any number of paths by induction.
\end{proof}

\begin{definition}[Trust evaluation]\label{def:TrustEval}
Let $G(V,E,\T)$ be a directed acyclic trust graph, and let a and b be two
nodes in $V$.
The trust evaluation between a and b is the trust aggregation over
{\em all} paths connecting $a$ to $b$. It is computed recursively by
aggregating (evaluating) the trust between the entity a and the
predecessors of b (except, potentially, $a$). Denote by $Pred(b)$
the predecessors of $b$ and by $p_i$ the elements of $Pred(b)\setminus\{a\}$.
The {\em trust evaluation} between $a$ and $b$ consists in first recursively
evaluating the trust over all paths $a\rightarrow \ldots \rightarrow p_i$, 
then applying the sequential aggregation over the paths 
$a\rightarrow p_i\rightarrow b$ and finally the parallel aggregation to the
results (and $(a\to b)$, if $(a\to b) \in E$).
\end{definition}
\begin{remark}
Note that since the predecessors of $b$ are distinct, after the sequential
aggregations all the resulting edges from $a$ to $b$ are distinct. They are
thus disjoint paths between $a$ and $b$ and parallel aggregation applies.
\end{remark}

\begin{remark}
In the above definition of trust evaluation we favor the evaluation from right
to left. As shown on the example below this gives in some sense prominence to
nodes close to the beginning of the path, that is nodes closer to the one asking
for an evaluation. This is illustrated on figure \ref{fig:distrib} where to
different strategies for the evaluation of trust are shown: on the left, one
with parallel, then sequential, aggregation; the other one with sequential, then
parallel, aggregation. 
\begin{figure}[htbp]
\begin{minipage}{\textwidth}
\begin{minipage}{0.3\textwidth}
From left to right, we get $f(a,g(f(b,c),d))$
\end{minipage}\hfill
\begin{minipage}{0.35\textwidth}
\begin{center}
\includegraphics[width=0.9\textwidth,keepaspectratio=true]{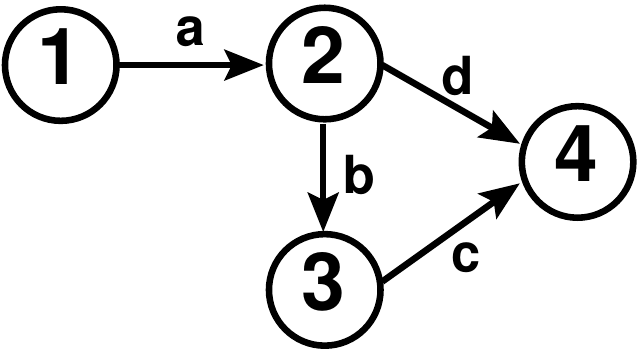}
\end{center}
\end{minipage}\hfill
\begin{minipage}{0.3\textwidth}
From right to left (definition \ref{def:TrustEval}), we get $g(f(a,b,c),f(a,d) )$
\end{minipage}
\end{minipage}
\caption{Two strategies for trust evaluation between node $1$ and
  $4$.}\label{fig:distrib}
\end{figure}
Would $f$ be distributive over $g$ we would get the same evaluation. As we will
see in section~\ref{sec:cycles}, this choice of evaluation can has an important
impact in the presence of cycles.
\end{remark}

The graph theoretic method proposed by \cite[\S 6.3]{Huang:2009:CTA} for
evaluating trust between two nodes in a DAG requires the
approximate solution of the Bounded Disjoint Paths problem, known to be
NP-Hard \cite{Reiter:1998:RAU}.
This algorithm has two steps: first an elimination of cycles via BDP, then a composition of sequential and then parallel aggregation.
We will show in the next section that our matrix algorithm produces the same
output on acyclic graphs. Moreover, in section~\ref{sec:cycles}, we will present
a variant of this algorithm, still with polynomial time bound, that directly
deals with generic graphs.

By storing the already evaluated relationships, the aggregation part of \cite[\S
6.3]{Huang:2009:CTA} can for instance compute the global trust in a graph with
$1000$ vertices and $250\,000$ edges in about $1$ minute on a standard laptop.
In the following, we propose to rewrite this algorithm in terms of linear
algebra~\cite{Orman:2010:wits}. Using sparse linear algebra the overall complexity will not change, and the analysis will be eased.
Now if the graph is close to complete, the trust matrix will be close to dense
and cache aware linear algebra block algorithms will be more suited.

In both cases, the linear algorithm will decompose the evaluation into
converging iterations which could be stopped before exact convergence in order
to get a good approximation faster. 

Furthermore, using this linear algebra point of view, we will we able to
generalize the algorithm to any directed graph.

\section{Matrix Powers Algorithm For Directed Acyclic Graphs}\label{sec:powers}
In the previous section, we presented the trust propagation scheme introduced
by \cite{Huang:2009:CTA}, which consists of using the parallel and sequential
trust aggregations for evaluating the trust between a network's entities.
Indeed, our matrix powers algorithm can be implemented with different trust
propagation schemes under one necessary condition: the transitivity
property of the (sequential and parallel) trust propagation formulas.
In this section, we propose a new algorithm for evaluating
trust in a network using the powers of the matrix of trust.
This algorithm uses techniques from graph connectivity
and communicability in networks
\cite{Estrada:2009:community}.

\subsection{Matrix and monoids of trust}
\begin{definition}\label{def:tmat}
Let $G=(V,E,\T)$ be a trust graph,
the {\em matrix of trust} of $G$, denoted by C, is the adjacency matrix
containing, for each node of the graph ,G the trust degrees of a node toward its
neighbors, $C_{ij} =\langle td(i,j),dtd(i,j),ud(i,j)\rangle$.
When there is no edge between $i$ and $j$, we choose $C_{ij}=\langle 0,0,1\rangle$ and, since every entity is fully confident in
itself, we also choose for all $i$: $C_{ii}=\langle 1,0,0\rangle.$
\end{definition}
\begin{definition}
Let $\T$ be the set $\T=\{\langle x,y,z\rangle\in [0,1]^3, x+y+z=1\}$, equipped with two
operations "+" and "." such that 
$\forall (\langle a,b,u\rangle,\langle c,d,v\rangle)\in \T^2$ we have:
$\langle a,b,u\rangle.\langle c,d,v\rangle=\langle ac+bd,ad+bc,1-ac-ad-bd-bc\rangle$,
and
$\langle a,b,u\rangle+\langle c,d,v\rangle=\langle a+c-ac,bd,(1-a)(1-c)-bd\rangle$.
We define as the {\em monoids of trust} the monoids $(\T, +, \langle 0,1,0\rangle)$
and $(\T,., \langle 1,0,0\rangle)$.
\end{definition}

$\langle 0,0,1\rangle$ is the absorbing element of "." in $\T$.
This justifies a posteriori our choice of representation for the
absence of an edge between two nodes in definition \ref{def:tmat}.

We can also see that the set $\T$ corresponds to trust degrees
$\langle td,dtd,ud\rangle$. In addition, the operations "." and "+" represent
respectively the sequential and parallel aggregations of trust, denoted $f$ and
$g$ in definitions~\ref{def:seqprop}~and~\ref{def:parprop}.

\begin{remark}\label{rq:distrib}
Note that "." is {\em not} distributive over "+".
This fact can prevent the use of block algorithms and fast matrix methods. 
Now if the simpler metric without distrust is used (i.e. distrust is fixed to
zero for every entity in the graph, and {\em will remain zero} all along our trust
algorithms), then "." becomes distributive over "+". 
Thus in the following
section we will present timings with or without taking distrust into
consideration.
\end{remark}

\subsection{$d$-aggregation of trust}
\begin{definition}[$d$-aggregation of trust]
For $d\in \N$, the {\em $d$-aggregation of trust} between two nodes
$A$ and $B$, in an acyclic trust graph, is the trust evaluation over
{\em all paths of length at most $d$}, connecting $A$ to $B$. It is
denoted $d\text{-agg}_{A,B}$.
\end{definition}

\begin{definition}[Trust vectors product]
Consider the directed trust graph $G=(V,E,\T)$ with trust matrix $C$.
Let  $\overrightarrow{C_{i*}}$ be the $i$-th row vector and
$\overrightarrow{C_{*j}}$ be the $j$-th column vector.
We define the product of $\overrightarrow{C_{i*}}$ by $\overrightarrow{C_{*j}}$ by: \[\overrightarrow{C_{i*}}.\overrightarrow{C_{*j}}=\sum^{k\ne j}_{k\in V}{C_{ik}.C_{kj}}\]
\end{definition}
Note that $C_{ii}=C_{jj}=\langle 1,0,0\rangle$ is the neutral element for ".". Therefore, our
definition differs from the classical dot product as we have
removed one of the $C_{ij}=C_{ii}\cdot C_{ij}=C_{ij}\cdot C_{jj}$, but then it
matches the $2$-aggregation:
\begin{lemma}\label{thm:cartesian}
The product  $\overrightarrow{C_{i*}}.\overrightarrow{C_{*j}}$ is the
$2$-aggregated trust between $i$ and~$j$.
\end{lemma}

\begin{proof}
We prove first that $C_{ik}.C_{kj}$ is the sequential
aggregation of trust between $i$ and $j$ throughout all the paths (of
length$\le 2$)  $i\to k\to j$ with $k \in V$.
Let $k$ be an entity in the network. There are two cases: whether $k$ is
one of the boundaries of the path or not.
The first case is:  $k=i$ or $k=j$:
\begin{itemize}
\item if $k=i$, then from the trust matrix definition \ref{def:tmat},
$C_{ii}=\langle 1,0,0\rangle\forall i$, thus we have
 $C_{ik}.C_{kj}=C_{ii}.C_{ij}=\langle 1,0,0\rangle.C_{ij}=C_{ij}.$
\item if $k=j$, then similarly
$C_{ik}.C_{kj}=C_{ij}.C_{jj}=C_{ij}.\langle 1,0,0\rangle=C_{ij}$
\end{itemize}
Therefore $C_{ik}.C_{kj}$ corresponds to the [~sequential aggregation
of ] trust between $i$ and $j$ throughout the path $(i,j)$ of length
$1$.
This is why in the product, we added the constraint $k\ne j$ in the
sum to avoid taking $C_{ij}$ twice into account.
Now the second case is:  $k\ne i$ and $k\ne j$:
\begin{itemize}
\item if $k$ belongs to a path of length $2$ connecting $i$ to $j$,
  then: $i$ trusts  $k$ with degree $C_{ik}\ne \langle 0,0,1\rangle$, and $k$ trusts
  $j$ with degree $C_{kj}\ne \langle 0,0,1\rangle$.
  From definition~\ref{def:seqprop}, $C_{ik}.C_{kj}$ corresponds to the
  sequential aggregation of trust between $i$ and $i$ throughout the
  path $i\to k\to j.$
\item If there is no path of length $2$ between $i$ and $j$ containing
  $k$, then we have $C_{ik}=\langle 0,0,1\rangle$ or $C_{kj}=\langle 0,0,1\rangle$, and
  thus $C_{ik}.C_{kj}=\langle 0,0,1\rangle$ is also the aggregation of
  trust between $i$ and $j$ on the path traversing the node $k$.
\end{itemize}
Finally, we can deduce that
$\overrightarrow{C_{i*}}.\overrightarrow{C_{*j}}=\sum^{k \ne j}_{k\in V}{C_{ik}.C_{kj}}$
corresponds to the parallel aggregation of trust between $i$ and $j$ using all paths of length $\le 2$, which is the $2$-aggregated trust between $i$ and $j$.
Note that the latter is equivalent to
$\overrightarrow{C_{i*}}.\overrightarrow{C_{*j}}=\sum_{
k\in Pred.(j)\setminus\{i\}
}{C_{ik}.C_{kj}}+C_{ij}$.
\end{proof}

\begin{definition}[Trust matrix product]\label{def:matprod}
Let $C_{(ij)}$ and $M_{(ij)}$ be two trust matrices. We define the matrix product  $N=C*M$ by: $\forall i,j\in \{1..n\}$
\begin{equation*}
N_{ij}= \begin{cases}
\overrightarrow{C_{i*}}.\overrightarrow{M_{*j}}=\displaystyle\sum^{k\ne j}_{k \in V}{C_{ik}.M_{kj}} & \text{ if $i\ne j$}
\\
\langle 1,0,0\rangle & \text{ otherwise}
\end{cases}
\end{equation*}
\end{definition}

\begin{lemma}\label{cor:twobytwo}
Let $(C_{ij})$ be the trust matrix of a network of entities,
whose elements belong to a trust graph $G$. The matrix
$M$ defined by: $M=C^2=C*C$ represents the $2$-aggregated trust
between all distinct entity pairs and the total trust of entities for
themselves. \end{lemma}

\begin{proof}
From definition \ref{def:matprod}, we have: 
$M_{ij}= \begin{cases}
\overrightarrow{C_{i*}}.\overrightarrow{C_{*j}}=\displaystyle\sum^{k\ne j}_{k \in V}{C_{ik}.C_{kj}} & \text{ if $i\ne j$}
\\
\langle 1,0,0\rangle & \text{ otherwise}
\end{cases}
$.\\
Thus, if $i=j$, then $M_{ii}=\langle 1,0,0\rangle$ as claimed. Otherwise, $i\ne j$ and
according to lemma \ref{thm:cartesian}, $M_{ij}$ the $2$-aggregated trust
between $i$~and~$j$.
\end{proof}

Now, according to definition \ref{def:TrustEval} of the trust evaluation in
a network, we can generalize lemma \ref{cor:twobytwo} to evaluate
trust using all paths of a given length:
\begin{theorem}\label{thm:length}
Let $G=(V,E,\T)$ be an acyclic trust graph with matrix of trust~$C$.
Then $C^d$ represents the $d$-aggregated trust
between all entity pairs in $V$.
\end{theorem}
\begin{proof}
We proceed by induction. Let $HR(d)$ be the hypothesis that $C^d$
represents the $d$-aggregated trust between all entity pairs in $V$.
Then $HR(2)$ is true from lemma \ref{cor:twobytwo}.
Now let us suppose that $HR(d)$ is true.
First if $i=j$, then $(d+1)\text{-agg}_{i,i}=\langle 1,0,0\rangle=C^{d+1}_{i,i}$ from
definition \ref{def:matprod}.
Second, definition \ref{def:TrustEval},
gives:
\begin{align*}
(d+1)\text{-agg}_{i,j} & = \sum_{k\in Pred.(j)\setminus
  \{i\}}{d\text{-agg}_{ik} . C_{kj}}+C_{ij} & \text{by
  def. \ref{def:TrustEval}}\\
 & = \sum_{k\in Pred.(j)\setminus
  \{i\}}{C^d_{ik} . C_{kj}}+C_{ij} & \text{by
  $HR(d)$}\\
 & =  \sum_{k\in Pred.(j)}{C^d_{ik} . C_{kj}} &
 \text{since $C_{ii}=\langle 1,0,0\rangle$}\\
 & =  \sum_{k\in V}^{k\ne j}{C^d_{ik} . C_{kj}} &\text{if $(ik)\notin E$},~
 \text{$C_{ik}=\langle 0,0,1\rangle$}
\end{align*}
Overall, $HR(d+1)$ is proven and induction proves the theorem.
\end{proof}

From this theorem, we immediately have that in an acyclic graph the
matrix powers must converge to a fixed point.

\begin{corollary}\label{cor:cv}
Let $G=(V,E,\T)$ be an acyclic trust graph with trust matrix $C$. The successive
application of the matrix powers of $C$ converges to a matrix $C^\ell$,
where $\ell$ is the size of the longest path in $G$. Which means that $C^\ell$
is a fixed point for the matrix powers:
$$\lim\limits_{n \to \infty} C^{n} = C^{l} $$
\end{corollary}

For the proof of this corollary, we need the following lemma.
\begin{lemma}\label{lem:maxpath}
Let $\lambda_{i,j}$ be the length of the largest path between $i$ and
$j$. Then either $\lambda_{i,j}=1$ or $\lambda_{i,j}=max_{k\in Pred.(j)\setminus\{i\}}(\lambda_{i,k})+1$.
\end{lemma}

\begin{proof}[Proof of corollary \ref{cor:cv}]
We prove the result by induction and consider the hypothesis $HR(d)$
with $d\ge 1$:
\[ \forall i,j \in V,~\text{such  that}~ \lambda_{i,j} \le d
~\text{and}~
\forall t \ge \lambda_{ij}, ~\text{then}~ C^t_{ij} = C^{\lambda_{i,j}}_{ij}
\]
We first prove the hypothesis for $d=1$:
Let $i$ and $j$ be such that the longest path between them is of
length $1$. This means that in this acyclic directed graph, there is
only one path between $i$ and $j$, the edge $i \to j$. Now from
definition \ref{def:TrustEval}, we have that
$\forall t, C^t_{ij}=  \sum_{k\in Pred.(j) \setminus
  \{i\}}{C^{t-1}_{ik}.C_{kj}}+C_{ij}$. However, $Pred.(j)\setminus \{i\} =
\emptyset$ so that $C^t_{ij}= C_{ij}$, for all
$t$. This proves $HR(1)$.

Now suppose that $HR(d)$ is true.
Let $i$ and $j$ be two vertices in $G(V,E,\T)$.
We have two cases.
First case: $\lambda_{ij} \le d$. Then $\lambda_{i,j} \le d+1$
and from the induction hypothesis, we have that
$C^t_{ij} = C^{\lambda_{ij}}_{ij} \forall t
\ge\lambda_{ij}$. Therefore $HR(d+1)$ is true for $i$ and $j$.
Second case: $\lambda_{ij} = d+1\ge 2$. Then we have $\forall u \ge 0$
    $C^{d+1+u}_{ij} =  \sum_{k\in Pred.(j)}{C^{d+u}_{ik}.C_{kj}}$.
    Now, from lemma \ref{lem:maxpath}, the maximum length of any path
    between $i$ and a predecessor of $j$ is
    $\lambda_{i,j}-1=d$. Therefore, from the induction hypothesis, we
    have that $C^{d+u}_{ik}=C^{\lambda_{ik}}_{ik}=C^d_{ik}$
    for all $k\in Pred.(j)$.
    Then $C^{d+1+u}_{ij} =  \sum_{k\in
      Pred.(j)}{C^{d}_{ik}.C_{kj}} = C^{d+1}_{ij}$ which proves
    the induction and thus the corollary.
\end{proof}

From the latter corollary, we now have an algorithm to compute the
trust evaluation between all the nodes in an acyclic trust network:
perform the trust matrix powering with the monoids laws up to the
longest path in the graph.
\begin{theorem}
Let $(C_{ij})$ be the trust matrix corresponding to an acyclic graph
with $n$ vertices and $\varphi$ edges whose longest path is of size $\ell$.
The complexity of the evaluation of the aggregated trust between all
entity pairs represented by this trust matrix is bounded by $O(n\cdot
\varphi \cdot \ell)$ operations.
\end{theorem}
\begin{proof}
$C$ is sparse with $\varphi$ non zero element. Thus multiplying $C$ by a vector
requires $O(\varphi)$ operations and computing $C \times C^i$
requires $O(n\varphi)$ operations. Then, theorem
\ref{thm:length} shows that $C^j$ for $j\geq \ell$ is the
$j$-aggregated trust between any entity pair. Finally, corollary
\ref{cor:cv} shows that $C^j=C^\ell$ as soon as $j\geq \ell$.
\end{proof}

The implementation of this algorithm took less than $1$ second
to perform an iteration ($C^2$) on the graph of section \ref{sec:metric}
with $1000$ vertices and $250K$ edges. And it needed less than $6$ seconds to return the final trust degrees. 

\section{Evaluation of trust in the presence of cycles}\label{sec:cycles}
The algorithm induced by theorem \ref{thm:length} works only for directed
acyclic graphs. Its advantage is thus restricted to the case when the distrust
is {\em not} taken into consideration: then block or sparse algorithms can
provide the BLAS\footnote{e.g. ATLAS~\cite{Whaley:2001:AEO},
  GotoBLAS~\cite{Goto:2008:blas}, MUMPS~\url{http://graal.ens-lyon.fr/MUMPS}
  etc.} linear algebra performance to trust evaluation.

Now, in the presence of cycles in a network, the matrix powers
algorithm will add the contribution of each edge of a cycle indefinitely.

Consider the graph of figure \ref{fig:TM12345}, with
$a,b,c,d$ the trust degrees corresponding to the links $1\overset{a}{\to}
2\overset{b}{\to} 3 \overset{c}{\to} 4 \overset{d}{\to} 2$.
Its trust matrix $C$ and applications of the matrix powers algorithm on this matrix are shown on figure \ref{fig:TM12345}, right.

\begin{figure}[htbp]
\begin{minipage}{0.35\textwidth}
\begin{center}
\includegraphics[width=0.9\textwidth,keepaspectratio=true]{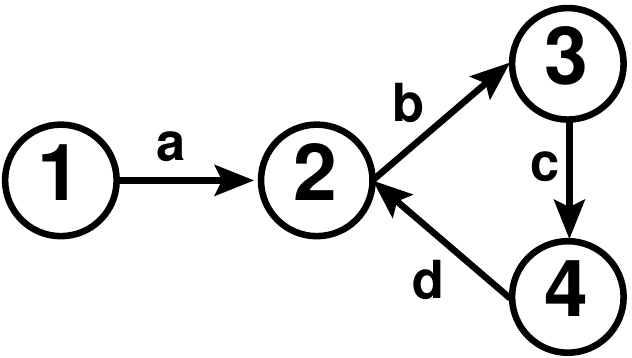}

\end{center}

\end{minipage}
\begin{minipage}{0.65\textwidth}\small
\hfill

\begin{tabular}{|c|c|c|c|}
\hline
1 & a & 0 & 0  \\ \hline
0 & 1 & b & 0  \\ \hline
0 & 0 & 1 & c  \\ \hline
0 & d & 0 & 1  \\ \hline
 \end{tabular}
\hfill
\begin{tabular}{|c|c|c|c|}
\hline
1 & a & a.b & 0  \\ \hline
0 & 1 & b & b.c  \\ \hline
0 & c.d & 1 & c  \\ \hline
0 & d & d.b & 1  \\ \hline
 \end{tabular}
\hfill
\begin{tabular}{|c|c|c|c|}
\hline
1 & a & a.b & a.b.c  \\ \hline
0 & 1 & b & b.c  \\ \hline
0 & c.d & 1 & c  \\ \hline
0 & d & d.b & 1  \\ \hline
 \end{tabular}\\

\hfill
\begin{tabular}{|c|c|c|c|}
\hline
1 & a+ a.b.c.d & a.b & a.b.c  \\ \hline
0 & 1 & b & b.c  \\ \hline
0 & c.d & 1 & c  \\ \hline
0 & d & d.b & 1  \\ \hline
 \end{tabular}
\hfill
\begin{tabular}{|c|c|c|c|}
\hline
1 & a +  a.b.c.d& (a+a.b.c.d).b & a.b.c  \\ \hline
0 & 1 & b & b.c  \\ \hline
0 & c.d & 1 & c  \\ \hline
0 & d & d.b & 1  \\ \hline
 \end{tabular}
\end{minipage}
\caption{Graph with one cycle and its trust matrix $C$ with $C^2$, $C^3$, $C^4$
  and $C^5$.}\label{fig:TM12345}
\end{figure}

For instance, the value:
\begin{equation*}\begin{split}
C^5_{1,3}&=0+C^4_{1,2}C_{2,3}+0=(C_{1,2}+C^3_{1,4}C_{4,2}) C_{2,3} \\
&= (C_{1,2}+(C^2_{1,3}C_{3,4})C_{4,2}) C_{2,3} =
(C_{1,2}+(C_{1,2}C_{2,3}) C_{3,4}) C_{4,2}) C_{2,3} \\
&= (a + a.b.c.d).b,
\end{split}\end{equation*}corresponds to
the aggregation on the paths $1\to 2\to 3$ and
$1\to 2\to 3\to 4\to 2\to 3$ linking 1 to 3. If we continue iterations
for $n > 5$, we find that the algorithm re-evaluates infinitely the trust on the
loop $3\to 4\to 2\to 3$ to yield 
$C^{2+3k}_{1,3}=(a+ C^{2+3{k-1}}_{1,3}.b.c.d).b$ with most probably an increase in uncertainty (e.g. from the many sequential aggregations and lemma \ref{lem:sequnincr}).

\subsection{Convergent iteration}
To solve this issue, we propose to change the matrix multiplication
procedure, so that {\em each edge will be used only once in the assessment
of a trust relationship}. For this, we use a memory matrix $R_{ij}$.
This stores, for each pair of nodes, all edges traversed to evaluate
their trust degree. Only the paths containing an edge not already traversed
to evaluate the trust degree are taken into account at the following iteration.
Therefore, the computation of $C^{\ell}_{ij}$ for $n\ge 1$, becomes that given
in algorithm \ref{alg:generic}. 

\begin{algorithm}[htb]
\caption{Matrix powers for generic network graphs}\label{alg:generic}
\begin{algorithmic}[1]
\REQUIRE An $n \times n$ matrix of trust $C$ of a generic directed trust graph.
\ENSURE Global trust in the network.

\STATE $\ell=2$;
\REPEAT
\FORALL{$(i,j) \in [1..n]^2$ with $i\neq j$}
\STATE $C_{ij}^{\ell}=\langle 0,0,1\rangle$; 
\STATE $R_{ij}^{\ell}  = \emptyset$;
\FOR{$k=1$ \TO $n$}
\STATE $t=C^{\ell-1}_{ik}.C_{kj}$;
\IF{$\left(t\neq \langle 0,0,1\rangle\right)$}
	\STATE $C_{ij}^{\ell} =  C_{ij}^{\ell} + t$;
      	\STATE $R_{ij}^{\ell}  = R_{ij}^{\ell} \bigcup R_{ik}^{\ell-1} \bigcup (k\to
j)$; // using a sorted list union
\ENDIF
\ENDFOR
\IFTHEN{$\left(\#R_{ij}^{\ell} \subset \#R_{ij}^{\ell-1}\right)$}{$C_{ij}^{\ell} = C_{ij}^{\ell-1} $;$R_{ij}^{\ell} = R_{ij}^{\ell-1} $;}
\ENDFOR
\UNTIL { $C^\ell == C^{\ell-1}$; ++$\ell$;}
\RETURN $C^{\ell}$;
\end{algorithmic}
\end{algorithm}

By doing this modification of the matrix multiplication we will obtain a
different trust evaluation. We have no guarantee that this new evaluation is
somewhat more accurate but intuitively as in any path an edge is considered only
once, cycles will not have a preponderate effect.
Furthermore, we recover an interesting fixed point.  

\begin{theorem}\label{thm:gencv}
Let $C$ be the trust matrix corresponding to a generic trust graph. Algorithm
\ref{alg:generic} converges to the matrix $C^\ell$ where $\ell$ is the longest
{\em acyclic path} between vertices.
\end{theorem}
\begin{proof}
Let $C^\ell$ be the evaluation of the $\ell$-aggregated trust between
all entity pairs after $\ell$ iterations where $\ell$ is the longest
acyclic path between vertices.
At this stage, for each pair $i,j$, all the edges belonging to a path
between $i$ and $j$ will be marked in $R^\ell_{ij}$. Therefore, no new
$t=C^{\ell}_{ik}.C_{kj}$ will be added to $C^{\ell+1}_{ij}$.
Conversely, at iteration $x <  \ell$, if there exist an acyclic path
between a pair $i,j$ of length greater than $x$, then it means that
there exists at least one edge $e$ not yet considered on a sub-path from, say, $u$ to $v$, of length $x$: $i \dashrightarrow u \dashrightarrow \overset{e}{\rightarrow} \dashrightarrow v \dashrightarrow j$.
Then $R^x_{uv}$ will be different from $R^{x-1}_{uv}$ and so will be $C_{uv}^x$ from $C_{uv}^{x-1}$.
\end{proof}

\begin{theorem}
Let $C$ be the trust matrix corresponding to a generic trust graph
with $n$ vertices and $\varphi$ edges whose longest path
between vertices is of size $\ell$.
The complexity of the global evaluation of all the paths between any
entity is bounded by $O(n^3\cdot\varphi\cdot\ell)$ operations.
\end{theorem}
\begin{proof}
Using algorithm \ref{alg:generic}, we see that the triple loop induces
$n^3$ monoid operations and $n^3$ merge of the sorted lists of edges.
A merge of sorted lists is linear in the number of edges, $\varphi$.
Then the overall iteration is performed at most $\ell$ times from
theorem \ref{thm:gencv}.
\end{proof}

By applying the new algorithm on the example of figure \ref{fig:TM12345},
we still obtain $C^5_{1,3}=(a+a.b.c.d).b$, but now $R^5_{1,3}=\{a,b,c,d\}$ and
thus no more contribution can be added to $C^{5+i}_{1,3}$.

A first naive dense implementation of this algorithm took about $1.3$ seconds to perform the first iteration ($C^2$) on the graph of section \ref{sec:metric} with $1000$ vertices and $250K$ edges. And it needed only $7$ iterations to return  the final trust degrees with high precision.  

\begin{remark}\label{rq:dominant}
More generally, the convergence will naturally be linked to the ratio of the
dominant and subdominant eigenvalues of the matrix. For instance on some random matrices it can be shown that this ratio is $O(\sqrt{n})$~\cite[Remark 2.19]{Goldberg:2000:mcap}, and this is what we had experimentally. 
Of course, for, e.g., PKI trust graphs, more studies on the structure of the typical network would have to be conducted.
\end{remark}

\subsection{Bounded evaluation of trust}
In practice, the evaluation of trust between two
nodes $A$ and $B$ need not consider all trust paths connecting $A$ to
$B$ for two reasons:
\begin{itemize}
\item  First, the mitigation is one of the trust properties, i.e. the
  trust throughout trust paths decreases with the length of the
  latter.
  Therefore after a certain length $L$, the trust on paths becomes
  weak and thus should have a low contribution in improving the trust
  degree after their parallel aggregation.
\item  Second, if at some iteration $n \geq 1$, we already obtained a
  high trust degree, then contributions of other paths will only be
  minor.
\end{itemize}
Therefore, it is possible to use the matrix powers algorithm with less
iterations and e.g. a threshold for the trust degree, in order to rapidly
compute a good approximation of the trust in a network. To determine the optimal
threshold, we have conducted hundreds of comparisons between results of each
iterations of our algorithm and the final trust degrees computed with Algorithm
1. We found that on average on $1K$-vertices random matrices, we needed 6
iterations to get an approximation of the trust degrees at 0.01, 
and only 7 iterations\footnote{From remark \ref{rq:dominant}, we see that would
  the product be a classical matrix product, $\sqrt{n}$ being about $32$, to get
  an approximation at $10^{-6}$ we would have needed on the order of
  $6/log_{10}{32} \approx 3.98 \leq 4$ iterations to converge.} 
(in $97\%$ of the cases) to achieve an error rate less than $10^{-6}$.

\section{Conclusion and remarks}
The actual public-key infrastructure models assume that the relationships
between the PKI entities are based on an absolute trust. However, several
risks are related to these assumptions when using PKI procedures.

In this paper, we have reduced the evaluation of trust between entities of a DAG
network to linear algebra. This gives a polynomial algorithm to asses the global
trust evaluation in  a network.
Moreover, depending on the sparsity of the considered graphs, this enables to
use adapted linear algebra methods. Also the linear algebra algorithm decomposes
the evaluation into converging iterations. These iterations can be terminated
earlier than convergence in order to get a good approximation faster. 
Finally this enabled us also to generalize the trust evaluation to any directed
graph, i.e. not necessarily acyclic, still with a polynomial complexity.

Further work include the determination of
the optimal number of iterations necessary to get a good approximation
of trust in minimal time. Small-world theory for instance could be of
help~\cite{Schnettler:2009:smallworld,Albert:1999:Diameter}.

Also, there is a restriction in the current trust models, which imposes
that Alice cannot communicate with Bob if there is no trust path between them.
However, this limitation can be overcome by using the reputation notion
jointly with the trust notion. If Alice can be sure that Bob has a good
reputation in its friends circle and vice versa, then they can extend their
trusty network and communicate safely \cite{Schiffner:2009:reputation}.

Finally, the choice of the implementation model is a crucial subject. One
approach is to adopt a centralized model, where all computation are done by a
unique trusty entity. Then the reliability of the system depends entirely on the
reliability of the trusted entity. This would typically be achieved by CAs.
Another approach is a distributed model, where the entities must contact each
others to share some trust degrees. This will enable each entity to evaluate the
(at least {\em local}) trust in its neighborhood. On the one hand, this can be
applied to large networks while preserving for each entity a low computational
cost. On the other hand, each entity might have only a limited view of the whole
network. Therefore, a dedicated Network Discovery Mechanism (NDM) is needed to
expand the entities trust sub-network. This NDM can be crucial to determine
the trust model safety \cite{Govindan:2012:surveytrust}.
Besides, the trust degrees could be a sensitive information. Therefore, the
join use of trust matrices and homomorphic cryptosystems enabling a
private computation of shared secret could be useful.

\bibliographystyle{abbrv}
\bibliography{pkibib}

\end{document}